\documentstyle[amsfonts,12pt]{article}
 
\title{}\date{}
\title{Solutions to the Master Equations Governing Fractional Vortices}
\author{Chang-Shou Lin\\Department of Mathematics,\\ National Taiwan University,\\ Taipei, Taiwan 10617, ROC\\ \\Gabriella Tarantello\\Dipartimento di Matematica\\Unversit\`{a} di
Roma ``Tor Vergata"\\Via della Ricerca Scientifica\\00133 Rome, Italy\\ \\Yisong Yang\\Institute of Contemporary Mathematics\\School of Mathematics\\Henan University\\
Kaifeng, Henan 475004, PR China\\and\\Department of Mathematics\\Polytechnic Institute of New York University\\Brooklyn, New York 11201, USA}
\newcommand{\bfR}{{\Bbb R}}
\newcommand{\bfC}{{\Bbb C}}
\newcommand{\ud}{\underline}

\def\XXint#1#2#3{{\setbox0=\hbox{$#1{#2#3}{\int}$}
 \vcenter{\hbox{$#2#3$}}\kern-.5\wd0}}

\newtheorem{oldtheorem}{Theorem}[section]
\newtheorem{oldassertion}[oldtheorem]{Assertion}
\newtheorem{oldproposition}[oldtheorem]{Proposition}
\newtheorem{oldremark}[oldtheorem]{Remark}
\newtheorem{oldlemma}[oldtheorem]{Lemma}
\newtheorem{olddefinition}[oldtheorem]{Definition}
\newtheorem{oldclaim}[oldtheorem]{Claim}
\newtheorem{oldcorollary}[oldtheorem]{Corollary}

\newenvironment{theorem}{\begin{oldtheorem}$\!\!\!${\bf.}}{\end{oldtheorem}}
\newenvironment{remark}{\begin{oldremark}$\!\!\!${\bf.}}{\end{oldremark}}

\newenvironment{lemma}{\begin{oldlemma}$\!\!\!${\bf.}}{\end{oldlemma}}

\newbox\qedbox
\newenvironment{proof}{\smallskip\noindent{\bf Proof.}\hskip \labelsep}%
                        {\hfill\penalty10000\copy\qedbox\par\medskip}

\newcommand{\dd}{\mbox{d}}
\newcommand{\ee}{\end{equation}}
\newcommand{\be}{\begin{equation}}\newcommand{\bea}{\begin{eqnarray}}
\newcommand{\eea}{\end{eqnarray}}
\newcommand{\ii}{\mbox{i}}\newcommand{\e}{\mbox{e}}
\newcommand{\pa}{\partial}\newcommand{\Om}{\Omega}
\newcommand{\vep}{\varepsilon}

\newcommand{\nn}{\nonumber}

\newcommand{\lm}{\lambda}

\begin{document}
\maketitle
\begin{abstract}
By means of variational methods, in this paper, we establish  sharp existence results for solutions of  the
master equations governing `fractional multiple vortices.'
In the doubly periodic situation, the conditions for existence are both necessary and sufficient and give the upper bounds for the vortex numbers
in terms of the size of the periodic cell domain.
In the planar situation, there is no restriction on the vortex numbers. In both situations, the solutions are uniquely determined by the prescribed locations
and the local winding numbers of the vortices.
\end{abstract}

\maketitle

\section{Introduction}
\setcounter{equation}{0}

Recall that the classical Abrikosov--Nielsen--Olesen vortices \cite{Ab,JT,NO} are static solutions of the Abelian Higgs model governed by the Hamiltonian
\be \label{a1.1}
{\cal H}=\frac12 F^2+\frac12|D\phi|^2+\frac\lm8(|\phi|^2-1)^2,
\ee
where $F=\pa_1 A_2-\pa_2 A_1$ is the magnetic field induced from the real-valued gauge potential field ${\bf A}=(A_1,A_2)$, $\phi$ is the complex-valued scalar Higgs field, $D\phi=
\nabla\phi-\ii {\bf A}\phi=(D_1\phi,D_2\phi)$ denotes the gauge-covariant derivative, $\lm>0$ is a coupling parameter, which satisfy the
Euler--Lagrange equations of (\ref{a1.1}) given by
\bea
\Delta\phi-2\ii {\bf A}\cdot\nabla\phi-\ii (\nabla\cdot{\bf A})\phi&=&|{\bf A}|^2\phi+\frac\lm2(|\phi|^2-1)\phi,\label{GL1}\\
\Delta {\bf A} -\nabla (\nabla\cdot {\bf A})&=&\frac\ii2(\overline{\phi}D\phi-\phi \overline{D\phi}),\label{GL2}
\eea
also known as the Ginzburg--Landau equations \cite{GL}. At the critical coupling $\lm=1$ separating type-I and type-II superconductivity, the
equations (\ref{GL1}) and (\ref{GL2}) over $\bfR^2$ are equivalent \cite{T2} to the self-dual or BPS (after the pioneering explorations of Bogomol'nyi \cite{Bo} and Prasad--Sommerfield \cite{PS}) vortex equations
\bea
D_1\phi\pm\ii D_2\phi&=&0,\label{glBPS1}\\
F\pm\frac12(|\phi|^2-1)&=&0.\label{glBPS2}
\eea
As a consequence, the solutions are all characterized precisely by the zero distribution
of the Higgs field $\phi$ \cite{JT,T1,T2} where superconductivity is absent. More precisely, the zeros of $\phi$, say $p_1,\cdots,p_N$, which may be arbitrarily prescribed, are
the centers of magnetic vortices. 
Solutions of such type are known as multiple vortices, and analytically (cf. \cite{JT}) they are governed by the master equation
\be \label{a1.6}
\Delta u=\e^u-1+4\pi\sum_{s=1}^N \delta_{p_s}(x).
\ee
Although the equation (\ref{a1.6}) seems to be rather simple, it is not integrable \cite{Schiff} and its understanding requires techniques \cite{JT,Ta1,WY,Ybook} based on functional analysis. Besides,
an equation of similar structure appears in conformal geometry \cite{Aubin,Chang1,Chang2,Chang3,CL} which has been an actively pursued subject in geometric analysis.

Physically, the solutions of (\ref{a1.6}) give rise to thinly formed vortex tubes in a superconductor, a phenomenon due to the Meissner effect, which is essential to a mechanism called monopole confinement \cite{Gr,ShY}.
Those facts provide some of the earliest clues \cite{Man1,Man2,Nambu,tH1,tH2} to the great puzzle of quark confinement.
Indeed, for quark confinement, Seiberg and Witten \cite{SW} showed how vortex tubes are generated in non-Abelian gauge field theory models through a dual Meissner effect
and how quark confinement is achieved through (color-charged) non-Abelian monopole condensation \cite{Gr,ShY}. Due to the interest and importance of non-Abelian vortices in monopole and quark confinement
problems, a burst of studies have been carried out and a wide range of nonlinear vortex equations
of rich features have been derived by theoretical physicists.
Here we mention the seminal papers of Marshakov and Yung \cite{MY}, Hanany and Tong \cite{HT}, and Auzzi,  Bolognesi,  Evslin,  Konishi, and Yung \cite{Auzzi},
and refer to the survey articles \cite{Eto-survey,Kon-survey1,Kon-survey2,ShY0,Tong} and references therein for  various relevant results in this active area.
Also we indicate in  \cite{Lieb-Yang,LY,LY2} a series of
existence and uniqueness results concerning the associated partial differential equation problems.

In the present paper, we aim at establishing the existence and uniqueness theorems for solutions of the master equations governing `fractional vortices' in the recent work of 
Eto, Fujimori,  Gudnason,  Konishi,  Nagashima, Nitta, Ohashi, and Vinci, \cite{Eto}. These equations may be regarded as extensions of the classical master equation (\ref{a1.6})
governing the Abelian Higgs vortices. We shall show that, unlike the equations arising in the Weinberg--Salam electroweak theory
\cite{BT,CL1,SYw1,Ta} and the Chern--Simons theory \cite{CY,Ta0,Ta1,Ybook}, the elegant structures of these master equations allow us to achieve a complete understanding of their solutions
in terms of explicitly stated conditions. In the next section, we describe the master equations
obtained in \cite{Eto}, to be studied here. We state our main existence and uniqueness results in \S 2 and in the subsequent sections, \S 3 -- \S 5, we give the proofs.
Our methods here are centered around a direct minimization procedure, an idea initiated in the recent work \cite{Lieb-Yang}.  We note that, while direct
minimization is convenient to implement for problems over a doubly periodic domain (due to compactness), it is not an easy task to carry on for problems over
the full plane (due to loss of compactness). In \S 5, we elaborate on such a problem and overcome the difficulty by establishing 
some coercive estimates for the action functionals. We shall see that the positive and negative parts of a trial function need to be treated separately and
the coercive lower bounds for them are rather different. We anticipate that this new method will be useful in dealing with other problems of similar difficulties.
In \S 7 and \S 8, we study and describe various degenerate situations.

\section{The master equations and existence theorems}
\setcounter{equation}{0}

The first master equation  obtained in \cite{Eto} is a scalar equation of the form
\be \label{1.1}
\Delta u=\lm (m|A|^2\e^{mu}+n|B|^2\e^{nu}-\xi)+4\pi\sum_{s=1}^N\delta_{p_s},
\ee
where $\xi>0$, $(A,B)\in\bfC^2$ is a constant vector which lies in the elliptical vacuum manifold (denoted by $M$) defined by the equation
\be \label{1.2}
m|A|^2+n|B|^2=\xi,
\ee
and the set of points $\{p_1,\cdots,p_N\}$, where the $p$'s are not necessarily distinct, describes the locations of $N$ vortices as in the Abelian Higgs case, (\ref{a1.6}). There are two problems of interest:
the solutions of (\ref{1.1}) defined over a doubly (gauge) periodic domain $\Om$, originally formulated in \cite{tH0}
and elaborated in concrete situations in \cite{CY,SYw1,WY}, and those over the full plane $\bfR^2$. In the latter situation, the vacuum manifold (\ref{1.2}) leads us
to the following natural boundary condition for $u$:
\be \label{1.3}
u\to0\quad\mbox{as }|x|\to\infty.
\ee

Our existence theorem for (\ref{1.1}) may be stated as follows.

\begin{theorem}\label{theorem2.1} Consider the multiple vortex equation (\ref{1.1}) with arbitrarily prescribed points 
$p_1,\cdots,p_N$, and assume that (\ref{1.2}) holds subject to the non-degeneracy condition $m,n\geq0$, and $\xi>0$. Then

(i) a solution exists over a doubly periodic domain $\Om$ if and only if
\be \label{2.4}
\frac{4\pi N}\lm<\xi|\Om|;
\ee

(ii) a solution over the full plane $\bfR^2$ subject to the boundary condition (\ref{1.3}) always exists and tends to zero exponentially fast;

(iii) in both cases, the solution $u$ is unique and satisfies
\be\label{bis}\begin{array}{ll}
&u<0,\\
& m|A|^2\|{\rm{\e}}^{mu}-1\|_1+n|B|^2\|{\rm{\e}}^{nu}-1\|_1=\frac{4\pi}\lm N,\end{array}
\ee
where $\|\cdot\|_1$ denotes the $L^1$-norm over $\Om$ or $\bfR^2$ respectively.

\end{theorem}

The second and third master equations obtained in \cite{Eto} are systems of two nonlinear equations of the form
\bea 
\Delta u_1&=&\lm_1 (m|A|^2\e^{mu_1}+|B|^2\e^{u_1+u_2}+|C|^2\e^{u_1-u_2}-\xi_1)+4\pi\sum_{s=1}^{N_1}\delta_{p_{1,s}},\label{4}\\
\Delta u_2&=&\lm_2 (|B|^2\e^{u_1+u_2}-|C|^2\e^{u_1-u_2}-\xi_2)+4\pi\sum_{s=1}^{N_2}\delta_{p_{2,s}},\label{5}
\eea
where $m=1$ or $2$, respectively,  $\lm_1,\lm_2,\xi_1,\xi_2>0$ are all fixed parameters, and $(A,B,C)\in \bfC^3$ satisfies 
\be \label{2.7}
m|A|^2+|B|^2+|C|^2=\xi_1,\quad |B|^2-|C|^2=\xi_2.
\ee 

The vanishing of either one of the parameters $B$ and $C$ makes the system of equations (\ref{4})--(\ref{5}) degenerate into the Abelian model (\ref{a1.6}), as will be discussed
in \S 7 and \S 8.

Thus, to capture the true new features of (\ref{4})--(\ref{5}), we shall discuss it first under the non-degeneracy conditions
\be \label{ndc}
B\neq0\quad\mbox{and}\quad C\neq0.
\ee
Notice that, for $C\neq0$,  we subject the sets of vortex points  to the condition
\be \label{7}
\{p_{2,1},\cdots,p_{2,{N_2}}\}\subset \{p_{1,1},\cdots,p_{1,N_1}\},
\ee
and impose the natural restriction
\be \label{N1N2}
N_2\leq N_1,
\ee
which ensure the regularity of the quantity $\e^{u_1-u_2}$.

From the vacuum manifold (\ref{2.7}), we see that for a solution $(u_1,u_2)$ of (\ref{4})--(\ref{5}) over $\bfR^2$, the natural boundary condition 
is expressed as follows:
\be \label{bc}
u_1,u_2\to0\quad\mbox{as }|x|\to\infty.
\ee 

On the other hand, the vacuum manifold (\ref{2.7}) plays no role as far as the solvability of the system of equations (\ref{4})--(\ref{5}) is concerned under periodic
boundary conditions with a given cell domain $\Om$. In this case, we identify necessary and sufficient conditions for the existence of a periodic solution that impose
upper bounds on the vortex numbers $N_1$ and $N_2$ (see (\ref{cd1}) and (\ref{cd2}) below), a feature common to other self-dual periodic vortex problems (cf. \cite{Ta1,Ybook}).
Actually, the vacuum constant enters to ensure suitable pointwise estimates on the solution, consistently with their physical interpretation.

More precisely, we prove 

\begin{theorem}\label{theorem2.2} Consider the master equations (\ref{4}) and (\ref{5}) for any integer $m>0$ subject to the non-degeneracy condition (\ref{ndc}).
Then

(i) a solution over a doubly periodic domain $\Om$ exists if and only if the conditions
\bea
4\pi\left(\frac{N_1}{\lm_1}+\frac{N_2}{\lm_2}\right)&<&(\xi_1+\xi_2)|\Om|,\label{cd1}\\
4\pi\left(\frac{N_1}{\lm_1}-\frac{N_2}{\lm_2}\right)&<&(\xi_1-\xi_2)|\Om|,\label{cd2}
\eea
 hold simultaneously;

(ii) a solution over the full $\bfR^2$ subject to the boundary condition (\ref{bc}) always exists provided (\ref{2.7}) holds;

(iii) in both cases above, when a solution exists, it must be unique;

(iv) furthermore, if we assume (\ref{2.7}) together with the condition
\be \label{2.15x}
\lm_2>\lm_1,\quad\xi_2>0,\quad\mbox{{and}}\quad m|A|^2<\min\left\{1,\frac2m\right\}\left(\frac{\lm_2}{\lm_1}-1\right)|C|^2,
\ee
then
\be \label{x2.16}
u_1\pm u_2<0,\quad u_1<0,
\ee 
and the following ``quantization" property hold:
\bea 
m|A|^2\|1-{\rm{\e}}^{mu_1}\|_1+2|B|^2\|1-{\rm{\e}}^{u_1+u_2}\|_1&=&4\pi\left(\frac{N_1}{\lm_1}+\frac{N_2}{\lm_2}\right),\label{q1}\\
m|A|^2\|1-{\rm{\e}}^{mu_1}\|_1+2|C|^2\|1-{\rm{\e}}^{u_1-u_2}\|_1&=&4\pi\left(\frac{N_1}{\lm_1}-\frac{N_2}{\lm_2}\right),\label{q2}
\eea
where $\|\,\|_1$ denotes the $L^1$-norm, over $\Om$ or $\bfR^2$, accordingly.
\end{theorem}

\begin{remark} The conditions (\ref{N1N2}) and (\ref{q2}) almost justify the necessity of the condition $\lm_2>\lm_1$.
\end{remark}

For convenience, degenerate cases where $B$ or $C$ will be allowed to vanish will be treated separately in latter sections.

\section{Proof of Theorem 2.1}
\setcounter{equation}{0}

We first consider (\ref{1.1}) over a doubly periodic domain $\Om$. Let $u_0$ be the unique periodic solution over $\Om$ of the problem:
\be 
\left\{\begin{array}{rll}\Delta u_0&=&-\frac{4\pi N}{|\Om|}+4\pi\sum_{s=1}^N\delta_{p_s},\\ \int_\Om u_0\,\dd x&=&0.\end{array}\right.
\ee
Thus, with $u=u_0+v$, the equation (\ref{1.1}) subject to periodic boundary condition becomes:
\be \label{4.2}
\left\{\begin{array}{ll}& \Delta v=\lm (m|A|^2\e^{mu_0 +mv}+n|B|^2\e^{nu_0+nv}-\xi)+\frac{4\pi N}{|\Om|},\quad x\in\Om,\\& v\mbox{ is doubly periodic over }\pa\Om.\end{array}\right.
\ee
Integrating (\ref{4.2}) over $\Om$, we have
\be \label{4.3}
\int_{\Om}\left(m|A|^2\e^{mu_0 +mv}+n|B|^2\e^{nu_0+nv}\right)\,\dd x=\xi|\Om|-\frac{4\pi N}{\lm}\equiv \eta>0,
\ee
which gives us the condition (\ref{2.4}) as a necessary condition for a solution to exist. We now show that (\ref{2.4}) is sufficient to ensure the existence
of a solution.

It is clear that (\ref{4.2}) is the Euler--Lagrange equation of the variational functional
\be \label{4.4}
I(v)=\int_\Om\left\{\frac12|\nabla v|^2+\lm\left(|A|^2\e^{mu_0+mv}+|B|^2\e^{nu_0+nv}\right)-\left(\lm\xi-\frac{4\pi N}{|\Om|}\right)v\right\}\,\dd x,
\ee
 over the Sobolev space 
\be 
H(\Om)=\{v\in W^{1,2}_{\mbox{loc}}(\bfR^2)\,|\, v \mbox{ is doubly periodic with periodic cell domain }\Om\},
\ee 
equipped with the standard scalar product and corresponding norm. 
To proceed further,  we decompose $H(\Om)=\bfR\oplus \dot{H}(\Om)$ with
\be \label{dc}
f=\underline{f}+\dot{f},\quad \underline{f}\in\bfR,\quad \dot{f}\in \dot{H}(\Om)\mbox{ satisfying }\int_\Om\dot{f}\,\dd x=0.
\ee 

By means of the Jensen inequality,
we have:
\bea \label{4.5}
I(v)&=&\int_{\Om}\left\{\frac12|\nabla\dot{v}|^2+\lm\left(|A|^2\e^{mu_0+m\dot{v}+m\ud{v}}+|B|^2\e^{nu_0+n\dot{v}+n\ud{v}}\right)\right\}\,\dd x-\lm\eta\ud{v}\nn\\
&\geq&\frac12\int_{\Om}|\nabla\dot{v}|^2\,\dd x+\lm|\Om|\left(|A|^2\e^{m\ud{v}}+|B|^2\e^{n\ud{v}}\right)-\lm\eta\ud{v}.
\eea
Since for $a\geq0,b\geq0, a+b>0, c>0, m\geq0,n\geq0$, and $am+bn>0$, the function
\be 
\sigma(t)=a\e^{m t}+b\e^{n t}-c t,
\ee
satisfies: $\sigma(t)\to\infty$ as $t\to\pm\infty$, we see from (\ref{4.5}) that $I(\cdot)$ is bounded from below and, by virtue of the Poincar\'{e} inequality,
it is also coercive in $H(\Om)$. By the weak lower semi-continuity of $I$ and its strict convexity, we conclude that $I$ attains the minimum value at a point that
is its unique critical point. So part (i)  of Theorem \ref{theorem2.1} is established.

Now consider (\ref{1.1}) over the full plane. In view of the condition (\ref{1.2}), it may be rewritten as follows:
\be \label{4.9}
\Delta u=\lm m|A|^2(\e^{mu}-1)+\lm n|B|^2(\e^{nu}-1)+4\pi\sum_{s=1}^N\delta_{p_s}.
\ee
We introduce the background function: 
\be \label{3.11}
u_0(x)=\sum_{s=1}^N\ln\left(\frac{|x-p_s|^2}{\mu+|x-p_s|^2}\right),
\ee
satisfying:
\be 
\Delta u_0= 4\pi\sum_{s=1}^N\delta_{p_s}-g_0,\quad g_0=4\mu\sum_{s=1}^N\frac1{(\mu+|x-p_s|^2)^2}.
\ee
Thus, setting $u=u_0+v$, we recast  (\ref{4.9}) into the equation
\be \label{4.10}
\Delta v=\lm m|A|^2(\e^{mu_0+mv}-1)+\lm n|B|^2(\e^{nu_0+nv}-1)+g_0,\quad x\in\bfR^2,
\ee
corresponding to the Euler--Lagrange equation of the functional
\bea \label{I0}
I(v)&=&\frac12\|\nabla v\|^2_2+\lm|A|^2\int_{\bfR^2}\left\{\e^{mu_0}(\e^{mv}-1-mv)+(\e^{mu_0}-1)mv\right\}\,\dd x\nn\\
&&\quad +\lm|B|^2\int_{\bfR^2}\left\{\e^{nu_0}(\e^{nv}-1-nv)+(\e^{nu_0}-1)nv\right\}\,\dd x+\int_{\bfR^2} g_0 v \,\dd x,
\eea 
well defined and of class $C^1$ over the Sobolev space $W^{1,2}(\bfR^2)$. Indeed,
the above functional is similar to that derived for the equation (\ref{a1.6}) in \cite{JT,T1}. Hence, using the method in \cite{JT,T1}, the existence and uniqueness of
a critical point of $I$ in the space $W^{1,2}(\bfR^2)$ can be established. Similarly, (\ref{bis}) follows for example by arguing as in Proposition 3.2.4 of \cite{Ta1}.
Here, we omit the details.

We note that the existence of a critical point of the functional (\ref{I0}) may also be established by our direct minimization method to be presented in \S 5.

\section{Proof of Theorem 2.2: doubly periodic case}
\setcounter{equation}{0}

We now consider the equations (\ref{4}) and (\ref{5}) over a doubly periodic domain $\Om$. Let $u_j^0$ ($j=1,2$) be the unique periodic solutions over $\Om$ of the equations:
\be 
\left\{\begin{array}{rll}\Delta u_j^0&=&-\frac{4\pi}{|\Om|} N_j +4\pi\sum_{s=1}^{N_j}\delta_{p_{j,s}}(x),\\ \int_\Om u^0_j\,\dd x&=&0,\quad j=1,2,\end{array}\right.
\ee
and set $u_3^0=u_1^0-u_2^0$. It is clear  that,  for $C\neq0$, the condition (\ref{7})  ensures the smoothness of the function $\e^{u_3^0}$.

Set $u_j=u_j^0+v_j$ ($j=1,2$). Then the system of equations (\ref{4}) and (\ref{5}) subject to periodic boundary conditions become:
\bea 
\Delta v_1&=&\lm_1 \left(m|A|^2\e^{mu_1^0+mv_1}+|B|^2\e^{u_1^0+u_2^0+v_1+v_2}+|C|^2\e^{u_3^0+v_1-v_2}-\xi_1\right)+\frac{4\pi}{|\Om|}{N_1},\quad\quad\label{2.2}\\
\Delta v_2&=&\lm_2 \left(|B|^2\e^{u^0_1+u^0_2+v_1+v_2}-|C|^2\e^{u_3^0+v_1-v_2}-\xi_2\right)+\frac{4\pi}{|\Om|}{N_2},\label{2.3}\\
 &&v_1, v_2 \mbox{ are doubly periodic on $\pa\Om$}.\nn
\eea
The problem above defines the Euler--Lagrange equations of the functional:
\bea \label{I1a}
&&I(v_1,v_2)=\nn\\
&&\int_\Om\left\{\frac1{2\lm_1}|\nabla v_1|^2+\frac1{2\lm_2}|\nabla v_2|^2+|A|^2\e^{mu_1^0+mv_1}+|B|^2\e^{u_1^0+u_2^0+v_1+v_2}+|C|^2\e^{u_3^0+v_1-v_2}\right.\nn\\
&&\quad\left.-\left(\xi_1-\frac{4\pi N_1}{|\Om|\lm_1}\right)v_1-\left(\xi_2-\frac{4\pi N_2}{|\Om|\lm_2}\right)v_2\right\}\,\dd x,\quad v_j\in H(\Om),\,j=1,2.
\eea
Again, on the basis of the Moser--Trudinger inequality (cf. \cite{Aubin,Fon}), we check that $I\in C^1(H(\Om)\times H(\Om))$. It is also easy to check the weak lower semi-continuity
of $I$.

On the other hand, integrating the  equations (\ref{2.2}) and (\ref{2.3}) over $\Om$, we obtain the natural constraints
\bea 
\int_\Om\left(\frac m2|A|^2\e^{mu_1^0+mv_1}+|B|^2\e^{u_1^0+u_2^0+v_1+v_2}\right)\,\dd x&=&\frac12(\xi_1+\xi_2)|\Om|-2\pi\left(\frac{N_1}{\lm_1}+\frac{N_2}{\lm_2}\right)\nn\\
&\equiv&\eta_1>0,\label{et1}\\
\int_{\Om}\left(\frac m2|A|^2\e^{mu_1^0+mv_1}+|C|^2\e^{u_3^0+v_1-v_2}\right)\,\dd x&=&\frac12(\xi_1-\xi_2)|\Om|-2\pi\left(\frac{N_1}{\lm_1}-\frac{N_2}{\lm_2}\right)\nn\\
&\equiv&\eta_2>0.\label{et2}
\eea 
Thus we deduce the necessity of the conditions (\ref{cd1}) and (\ref{cd2}) for the existence of a solution once we take into account of (\ref{ndc}).

Now assume that the conditions (\ref{cd1}) and (\ref{cd2}) are fulfilled. We show that (\ref{4}) and (\ref{5}) will have a solution.

To this purpose we use the decomposition (\ref{dc}) to obtain the following:
\bea \label{I2a}
&&I(v_1,v_2)-
\int_\Om\left\{\frac1{2\lm_1}|\nabla \dot{v}_1|^2+\frac1{2\lm_2}|\nabla \dot{v}_2|^2+|A|^2\e^{mu_1^0+mv_1}\right\}\,\dd x\nn\\
&&=|B|^2\int_\Om \e^{u_1^0+u_2^0+\ud{v}_1+\ud{v}_2+\dot{v}_1+\dot{v}_2} \,\dd x -\eta_1 (\ud{v}_1+\ud{v}_2)\nn\\
&&\quad +|C|^2\int_\Om \e^{u_3^0+\ud{v}_1-\ud{v}_2+\dot{v}_1-\dot{v}_2}\,\dd x
-\eta_2(\ud{v}_1-\ud{v}_2)\nn\\
&&\geq |B|^2 |\Om|\e^{\ud{v}_1+\ud{v}_2}-\eta_1 (\ud{v}_1+\ud{v}_2)
+|C|^2 |\Om|\e^{\ud{v}_1-\ud{v}_2}
-\eta_2(\ud{v}_1-\ud{v}_2),
\eea
where again we have used the Jensen inequality.

At this point, by the non-degeneracy condition (\ref{ndc}) and the fact that $\eta_1>0$ and $\eta_2>0$ (see (\ref{et1}) and (\ref{et2})), as above we deduce that $I$ is
coercive and bounded from below in $H(\Om)\times H(\Om)$. By weak lower semi-continuity, we see that $I$ attains its minimum value. Since $I$ is strictly convex, the minimum
point gives the unique critical point of $I$, and the first part of Theorem \ref{theorem2.2} is proved.

\section{Proof of Theorem 2.2: planar case}
\setcounter{equation}{0}

In this section, we study the equations (\ref{4}) and (\ref{5}) over the entire plane $\bfR^2$ subject to the topological boundary condition (\ref{bc}). 
As usual, $\|\cdot\|_p$ denotes the norm of $L^p(\bfR^2)$ and $B_r$ denotes the disk in $\bfR^2$ centered at the origin with radius $r>0$. We will
continue to assume the condition (\ref{7}) imposed on the sets of vortex points.
Our goal is to establish an existence and uniqueness theorem as in the doubly periodic situation.

To this end, we introduce the background functions:
\bea 
u_j^0(x)&=&\sum_{s=1}^{N_j}\ln\left(\frac{|x-p_{j,s}|^2}{\mu+|x-p_{j,s}|^2}\right),\quad h_j(x)=\e^{u^0_j(x)},\quad j=1,2,\label{5u}\\
g_j(x)&=& 4\mu\sum_{s=1}^{N_j}\frac1{(\mu+|x-p_{j,s}|^2)^2},\quad j=1,2.\label{5g}
\eea

We now set
\bea 
H_1(x)&=&h_1(x)h_2(x),\quad H_2(x)=\frac{h_1(x)}{h_2(x)},\label{x5.3}\\
 G_1(x)&=&\frac1{\lm_1}g_1(x)+\frac1{\lm_2}g_2(x),\quad G_2(x)=\frac1{\lm_1}g_1(x)-\frac1{\lm_2}g_2(x),\label{x5.4}
\eea
which are all smooth functions.

Setting $u_j=u_j^0+v_j$ ($j=1,2$), we see that the equations (\ref{4}) and (\ref{5})
subject to the boundary condition (\ref{bc}) become the following boundary value problem over $\bfR^2$:
\bea 
\Delta v_1&=& m\lm_1|A|^2 (h_1^m\e^{mv_1}-1)+\lm_1|B|^2 (H_1\e^{v_1+v_2}-1)
+\lm_1|C|^2 (H_2\e^{v_1-v_2}-1)+g_1,\nn\\ \label{eq1}\\
\Delta v_2&=&\lm_2|B|^2 (H_1\e^{v_1+v_2}-1)-\lm_2|C|^2(H_2\e^{v_1-v_2}-1)+g_2,\label{eq2}\\
v_1&\to&0,\quad v_2\to0,\quad |x|\to\infty.\label{bc1}
\eea

As in the doubly periodic situation, we pursue a solution of the problem (\ref{eq1})--(\ref{bc1}) by a minimization procedure. For this purpose, we observe that
the problem enjoys a variational formulation with the action functional:
\bea\label{I1}
&&I(v_1,v_2)=\nn\\
&&\frac1{2\lm_1}\|\nabla v_1\|^2_2+\frac1{2\lm_2}\|\nabla v_2\|^2_2+|A|^2\int_{\bfR^2}(h_1^m[\e^{mv_1}-1-mv_1]+[h_1^m-1] mv_1)\,\dd x\nn\\
&&+|B|^2\int_{\bfR^2}(H_1[\e^{v_1+v_2}-1-(v_1+v_2)]+[H_1-1][v_1+v_2])\,\dd x\nn\\
&&+|C|^2\int_{\bfR^2}(H_2[\e^{v_1-v_2}-1-(v_1-v_2)]+[H_2-1][v_1-v_2])\,\dd x\nn\\
&&+\int_{\bfR^2}\left(\frac1{\lm_1}g_1 v_1+\frac1{\lm_2}g_2 v_2\right)\,\dd x,\label{I}
\eea
where $v_1,v_2\in W^{1,2}(\bfR^2)$.

Notice that,
\be\label{5.8bis}
\begin{array}{ll}&0\leq h_j<1,\quad 0\leq H_j<1;\quad (1-h_j^m), (1-H_j)\in L^p(\bfR^2),\,\forall p>1;\\
& g_j\in L^1(\bfR^2)\cap L^\infty(\bfR^2);\quad j=1,2.\end{array}
\ee
So the functional (\ref{I}) is well defined and of class $C^1$ in $W^{1,2}(\bfR^2)\times W^{1,2}(\bfR^2)$.

Moreover, exactly as for the Abelian case, elliptic estimates and standard regularity theory ensures that a critical point of $I$ 
in $W^{1,2}(\bfR^2)\times W^{1,2}(\bfR^2)$ is a solution of (\ref{eq1})--(\ref{bc1}). Furthermore, it can easily be checked that $I$ is weakly lower semi-continuous and
strictly convex. So it allows at most one critical point.

Fix ${\lm}>\max\{\lm_1,\lm_2\}>0$ and rewrite (\ref{I}) as follows:
\bea \label{I2}
&&I(v_1,v_2)=\nn\\
&&\left(\frac1{\lm_1}-\frac1\lm\right)\frac12\|\nabla v_1\|^2_2+\left(\frac1{\lm_2}-\frac1\lm\right)\frac12\|\nabla v_2\|^2_2+\frac1{4\lm}\left(\|\nabla(v_1+v_2)\|^2_2+\|\nabla(v_1-v_2)\|_2^2\right)\nn\\
&&+|A|^2\int_{\bfR^2}(h_1^m[\e^{mv_1}-1-mv_1]+[h_1^m-1] mv_1)\,\dd x\nn\\
&&+|B|^2\int_{\bfR^2}(H_1[\e^{v_1+v_2}-1-(v_1+v_2)]+[H_1-1][v_1+v_2])\,\dd x\nn\\
&&+|C|^2\int_{\bfR^2}(H_2[\e^{v_1-v_2}-1-(v_1-v_2)]+[H_2-1][v_1-v_2])\,\dd x\nn\\
&&+\int_{\bfR^2}\left(\frac1{\lm_1}g_1 v_1+\frac1{\lm_2}g_2 v_2\right)\,\dd x\nn\\
&=&\left(\frac1{\lm_2}-\frac1\lm\right)\frac12\|\nabla v_2\|^2_2+\left(\frac1{\lm_1}-\frac1\lm\right)\frac1{m^2}J_0(mv_1)
+\frac1{2\lm}\left(J_1(v_1+v_2)+J_2(v_1-v_2)
\right)\nn\\
&&\quad +\int_{\bfR^2}\left(\frac1{\lm_1}g_1 v_1+\frac1{\lm_2}g_2 v_2\right)\,\dd x,
\eea
where 
\bea
J_0(v)&=&\frac12\|\nabla v\|^2_2
+m^2\left(\frac1{\lm_1}-\frac1\lm\right)^{-1}|A|^2\int_{\bfR^2}(h_1^m[\e^v-1-v]+[h_1^m-1]v)\,\dd x,\quad\quad
\label{J0}\\
J_1(v)&=&\frac12\|\nabla v\|^2_2+2\lm|B|^2\int_{\bfR^2}(H_1[\e^v-1-v]+[H_1-1]v)\,\dd x,
\label{J1}\\
J_2(v)&=&\frac12\|\nabla v\|^2_2+2\lm|C|^2\int_{\bfR^2}(H_2[\e^v-1-v]+[H_2-1]v)\,\dd x.\label{J2}
\eea
These functionals may collectively be expressed in the form:
\be \label{Jv}
J(v)=\frac12\|\nabla v\|^2_2+\alpha^2\int_{\bfR^2}\left(H_{0,\mu}(\e^v-1-v)+(H_{0,\mu}-1)v\right)\,\dd x,
\ee
where $\alpha>0$ is a constant and
\be 
u_{0,\mu}(x)=\sum_{s=1}^N \ln\left(\frac{|x-p_s|^2}{\mu+|x-p_s|^2}\right),\quad H_{0,\mu}=\e^{u_0,\mu},
\ee
with $p_1,\cdots,p_N\in\bfR^2$ fixed points repeated  accounting to their multiplicities. We have $0\leq H_{0,\mu}<1$ and $1-H_{0,\mu}\in L^2(\bfR^2)$. Besides, for
\be 
R_0>3\max\{|p_s|\,|\, s=1,\cdots,N\},
\ee
there exists some $c_0=c_0(R_0)>1$, such that for any $\mu\geq1$ we have:
\bea 
H_{0,\mu}(x)&\geq&\frac1{c_0\mu^N},\quad  |x|\geq R_0,\label{1-1}\\
0\leq H_{0,\mu}(x)&\leq&\frac{c_0}{\mu^N},\quad |x|\leq R_0,\label{1-2}\\
\int_{B_{R_0}} H_{0,\mu}(x)\,\dd x&\geq&\frac1{c_0\mu^N}.\label{1-3}
\eea
%where $B_R$ denotes the disk in $\bfR^2$ centered at the origin and of radius $R$.

\begin{lemma}\label{lemma5.1} For the functional defined in (\ref{Jv}), there is a constant $\mu_\alpha>1$ such that for $\mu>\mu_\alpha$ there holds:
\be \label{Jv2}
J(v)\geq\frac14\|\nabla v\|^2_2+\frac\alpha2\|v_-\|+\frac{\alpha^2 b_0}{\mu^N}\|v_+\|^2-C_\mu,\quad v\in W^{1,2}(\bfR^2),
\ee
where $v^+=\max\{ v,0\}$ and $v^-=\max\{-v,0\}$ are the positive and negative parts of $v$, respectively, $b_0>0$ is a constant 
independent of $\mu$ and $C_\mu>0$ is some constant depending on $\mu$.
\end{lemma}
\begin{proof} Since $J(v)=J(v^+)+J(v^-)$, we can split the proof  into two cases.

Case 1. $v\leq0$ a.e. in $\bfR^2$.

Recall the Sobolev inequality
\be \label{Sobolev}
\int_{\bfR^2} v^4\,\dd x\leq 2\left(\int_{\bfR^2} v^2\,\dd x\right)\left(\int_{\bfR^2}|\nabla v|^2\,\dd x\right),\quad v\in W^{1,2}(\bfR^2).
\ee
Thus, we have
\bea \label{u4}
\left(\int_{\bfR^2} v^2\,\dd x\right)^2&=&\left(\int_{\bfR^2}\frac{|v|}{1+|v|}(1+|v|)|v|\,\dd x\right)^2\nn\\
&\leq&2\int_{\bfR^2}\frac{v^2}{(1+|v|)^2}\,\dd x\int_{\bfR^2}(v^2+v^4)\,\dd x.
\eea
Inserting (\ref{Sobolev}) into (\ref{u4}), we have
\be\label{5.14}
\|v\|^2_2\leq 2\left(1+2\|\nabla v\|^2_2\right)\int_{\bfR^2}\frac{v^2}{(1+|v|)^2}\,\dd x,\quad v\in W^{1,2}(\bfR^2).
\ee
Besides, from the well-known inequality:
\be\label{5.15}
1-\e^{-t}\geq \frac t{1+t},\quad\forall t\geq0,
\ee
we deduce that,
\be \label{5.16}
\e^{-t}-1+t=\int_0^t(1-\e^{-\tau})\,\dd\tau\geq\int_0^t\frac\tau{1+\tau}\,\dd\tau
\geq\frac{t^2}{2(1+t)},\quad\forall t\geq0.
\ee
Hence,  for $v\leq0$, we obtain the lower bound:
\bea
H_{0,\mu}(\e^v-1-v)+(H_{0,\mu}-1)v&=&(H_{0,\mu}-1)(\e^v-1)+(\e^v-1-v)\nn\\
&\geq&\e^v-1-v\nn\\
&\geq&\frac12\frac{v^2}{1+|v|}.
\eea
Consequently, using (\ref{5.14}), we have
\bea 
J(v)&\geq&\frac12\|v\|^2_2+\frac{\alpha^2}2\int_{\bfR^2}\frac{v^2}{1+|v|}\,\dd x\nn\\
&\geq&\frac12\|v\|^2_2+\frac{\alpha^2}2\int_{\bfR^2}\frac{v^2}{(1+|v|)^2}\,\dd x\nn\\
&\geq&\frac12\left(\|\nabla v\|^2_2+\frac{\alpha^2}2\frac{\|v\|^2_2}{(1+2\|\nabla v\|^2_2)}\right)\nn\\
&\geq&\frac14\|\nabla v\|^2_2+\frac14\left(\|\nabla v\|^2_2+\frac{\alpha^2\|v\|^2_2}{1+\|\nabla v\|^2_2}\right).
\eea
Therefore, we can minimize the function
\be 
\varphi(t)=t+\frac{\alpha^2\| v\|^2_2}{1+t},\quad t>0,
\ee
to find $\varphi(t)\geq 2\alpha\|v\|_2-1$, and conclude:
\be \label{5.25}
J(v)\geq\frac14 \|\nabla v\|^2_2+\frac\alpha2\|v\|_2-1,\quad v\leq0.
\ee

Case 2: $v\geq0$ a.e. in $\bfR^2$.

First, note that 
\be 
\int_{\bfR^2}\left(H_{0,\mu}(\e^v-1-v)+(H_{0,\mu}-1)v\right)\,\dd x\geq \int_{\bfR^2}\left(\frac12H_{0,\mu} v^2+(H_{0,\mu}-1)v\right)\,\dd x,
\ee
and from (\ref{1-1}), we get
\bea \label{5.27}
&&\int_{|x|\geq R_0}\left(\frac12H_{0,\mu} v^2+(H_{0,\mu}-1)v\right)\,\dd x\nn\\
&&\geq\frac1{2c_0\mu^N}\|v\|^2_{L^2(\bfR^2\setminus B_{R_0})}-\|1-H_{0,\mu}\|_{L^2(\bfR^2\setminus B_{R_0})}\|v\|_{L^2(\bfR^2\setminus B_{R_0})}\nn\\
&&\geq\frac1{4c_0\mu^N}\|v\|^2_{L^2(\bfR^2\setminus B_{R_0})}-C_{1,\mu},
\eea
for a suitable constant $C_{1,\mu}>0$.

On the other hand, if we set
\be 
M=\frac1{|B_{R_0}|}\int_{B_{R_0}} v\,\dd x,\quad v_1=v-M,
\ee
we find:
\bea 
&&\int_{B_{R_0}}\left(\frac12 H_{0,\mu} v^2+(H_{0,\mu}-1)v\right)\,\dd x\nn\\
%&&=\frac12\int_{B_{R_0}} H_{0,\mu} v_1^2\,\dd x+M\int_{B_{R_0}} v_1\,\dd x+\frac{M^2}2\int_{B_{R_0}}H_{0,\mu}\,\dd x\nn\\
%&&\quad+\int_{B_{R_0}} (H_{0,\mu}-1)v_1\,\dd x+M\int_{B_{R_0}} (H_{0,\mu}-1)\,\dd x\nn\\
&&=\frac12\int_{B_{R_0}} H_{0,\mu} v_1^2\,\dd x+\int_{B_{R_0}} H_{0,\mu}v_1\,\dd x+\frac{M^2}2\int_{B_{R_0}}H_{0,\mu}\,\dd x\nn\\
&&\quad+M\left(\int_{B_{R_0}} H_{0,\mu} v_1\,\dd x+\int_{B_{R_0}} (H_{0,\mu}-1)\,\dd x\right)\nn\\
&&\geq\frac12\int_{B_{R_0}} H_{0,\mu} v_1^2\,\dd x+\int_{B_{R_0}} H_{0,\mu}v_1\,\dd x+\frac{M^2}4\int_{B_{R_0}}H_{0,\mu}\,\dd x\nn\\
&&\quad -\left(\int_{B_{R_0}} H_{0,\mu}\,\dd x\right)^{-1}\left(\int_{B_{R_0}}(H_{0,\mu} v_1+H_{0,\mu}-1)\,\dd x\right)^2\nn\\
&&\geq\frac12\int_{B_{R_0}} H_{0,\mu} v_1^2\,\dd x+\int_{B_{R_0}} H_{0,\mu}v_1\,\dd x+\frac{M^2}4\int_{B_{R_0}}H_{0,\mu}\,\dd x\nn\\
&&\quad -\left(\int_{B_{R_0}} H_{0,\mu}\,\dd x\right)^{-1}\left(\left[\int_{B_{R_0}}H_{0,\mu} v_1\,\dd x\right]^2-2\left[\int_{B_{R_0}}H_{0,\mu} v_1\,\dd x\right]\times\right.\nn\\
&&\quad\times\left.\left[
\int_{B_{R_0}}(1-H_{0,\mu})\,\dd x\right]+\left[\int_{B_{R_0}}(1-H_{0,\mu})\,\dd x\right]^2\right)\nn\\
&&\geq\frac{M^2}4\int_{B_{R_0}} H_{0,\mu}\,\dd x+\left(\int_{B_{R_0}} H_{0,\mu} v_1\,\dd x\right)\left(2\left[\int_{B_{R_0}} H_{0,\mu}\,\dd x\right]^{-1}-1\right)\nn\\
&&\quad -\frac12\left(\int_{B_{R_0}} H_{0,\mu}\,\dd x\right)^{-1}\left(\int_{B_{R_0}} H_{0,\mu}v_1\,\dd x\right)^2 -C_{2,\mu}\nn\\
&&\geq\frac{M^2}4\int_{B_{R_0}} H_{0,\mu}\,\dd x-c_\mu\|v_1\|_{L^2(B_{R_0})}\nn\\
&&\quad-\frac12\left(\int_{B_{R_0}} H_{0,\mu}\,\dd x\right)^{-1}\left(\int_{B_{R_0}} H_{0,\mu}^2\,\dd x\right)\|v_1\|^2_{L^2(B_{R_0})} -C_{2,\mu},
\eea
for some suitable constants $c_\mu, C_{2,\mu}>0$, after using the inequality
\be 
\int_{B_{R_0}} H_{0,\mu}v_1^2\,\dd x \geq\left(\int_{B_{R_0}} H_{0,\mu}\,\dd x\right)^{-1}\left(\int_{B_{R_0}} H_{0,\mu}v_1\,\dd x\right)^2.
\ee
As a consequence, using (\ref{5.27}), for $v\geq0$, we have
\bea \label{Jv1}
J(v)&\geq&\frac12\|\nabla v\|^2_2+\alpha^2 \left(\frac14\int_{B_{R_0}} H_{0,\mu}\,\dd x\right)M^2\nn\\
&&-\alpha^2\left(\int_{B_{R_0}} H_{0,\mu}\,\dd x\right)^{-1}\left(\int_{B_{R_0}} H_{0,\mu}^2\,\dd x\right)\|v_1\|^2_{L^2(B_{R_0})}\nn\\
&& +\frac{\alpha^2}{4c_0\mu^N}\|v\|^2_{L^2(\bfR^2\setminus B_{R_0})}
-C_{3,\mu},
\eea
for some constant $C_{3,\mu}>0$. Now notice that
\be 
\left(\int_{B_{R_0}} H_{0,\mu}\,\dd x\right)^{-1}\left(\int_{B_{R_0}} H_{0,\mu}^2\,\dd x\right)\leq \|H_{0,\mu}\|_{L^\infty(B_{R_0})}\leq \frac{c_0}{\mu^N},
\ee
in view of (\ref{1-2}). On the other hand, by the Poincar\'{e} inequality, we have
\be \label{p1}
\|v_1\|^2_{L^2(B_{R_0})}\leq a_0 R^2_0\|\nabla v_1\|^2_{L^2(B_{R_0})},
\ee
which leads us to obtain:
\be \label{p2}
\|v\|^2_{L^2(B_{R_0})}\leq a_0 R^2_0\left(\|\nabla v\|^2_{L^2(B_{R_0})}+M^2\right).
\ee
Thus, by (\ref{1-2}), (\ref{1-3}), (\ref{Jv1}), (\ref{p1}), and (\ref{p2}), we can find constants $b_0, b_1>0$, depending on $R_0$ but independent of $\mu$, such that
\be \label{5.35}
J(v)\geq\left(\frac12-\frac{\alpha^2 b_0}{\mu^N}\right)\|\nabla v\|^2_2+\frac{\alpha^2 b_1}{\mu^N}\|v\|^2_2-C_\mu,\quad v\geq0.
\ee

Summarizing Cases 1 and 2 above, i.e., (\ref{5.25}) and (\ref{5.35}), we can find $\mu_\alpha>1$ sufficiently large such that (\ref{Jv2}) holds whenever $\mu\geq\mu_\alpha$.
\end{proof}

We are now ready to prove that the functional $I$ defined in (\ref{I1}) or (\ref{I2}) is coercive. In fact, applying Lemma \ref{lemma5.1}
to the functionals $J_\ell$, $\ell=0,1,2$, defined in (\ref{J0})--(\ref{J2}), we can find constants $a,b>0$ (independent of $\mu$) and $c_\mu>1$, such that for $\mu>0$ sufficiently large
and $N=\max\{N_1+N_2,mN_1\}$, we have
\bea
I(v_1,v_2)&\geq& a\left(\|\nabla v_1\|_2^2+\|\nabla v_2\|^2_2+|A|\|{v_1}_-\|_2+|B|\|(v_1+v_2)_-\|_2+|C|\|(v_1-v_2)_-\|_2\right)\nn\\
&&+\frac b{\mu^N}\left(|A|^2\|{v_1}_+\|^2_2+|B|^2\|(v_1+v_2)_+\|_2^2+|C|^2\|(v_1-v_2)_+\|^2_2\right)\nn\\
&&+\int_{\bfR^2}\left(\frac1{\lm_1}g_1 v_1+\frac1{\lm_2}g_2 v_2\right)\,\dd x -c_\mu.\label{Iv12}
\eea

Recalling (\ref{x5.4}), for $B\neq0, C\neq0$, we can estimate
\bea
\int_{\bfR^2}\left(\frac1{\lm_1}g_1 v_1+\frac1{\lm_2}g_2 v_2\right)\,\dd x&=&\frac12\int_{\bfR^2} G_1 (v_1+v_2)\,\dd x+\frac12\int_{\bfR^2} G_2 (v_1-v_2)\,\dd x\nn\\
&\geq&-\frac12\|G_1\|_2\|v_1+v_2\|_2 -\frac12\|G_2\|\|v_1-v_2\|_2\nn\\
&\geq&-\frac c\mu\left(\|v_1+v_2\|_2+\|v_1-v_2\|_2\right)\nn\\
&\geq&-\frac c\mu\left(\|(v_1+v_2)_-\|_2+\|(v_1-v_2)_-\|_2\right)\nn\\
&&-\frac{\vep}{\mu^N}\left(\|(v_1+v_2)_+\|_2^2+\|(v_1-v_2)_+\|_2^2\right)-C_{\vep,\mu},\label{5.40}
\eea
where $\vep>0$ is arbitrarily small. Incidentally, also notice that in the case $C=0$ but $A\neq0, B\neq0$, we can estimate
\bea
\int_{\bfR^2}\left(\frac1{\lm_1}g_1 v_1+\frac1{\lm_2}g_2 v_2\right)\,\dd x&=&\int_{\bfR^2}\left(\left[\frac1{\lm_1}g_1-\frac1{\lm_2}g_2\right] v_1+\frac1{\lm_2}g_2[v_1+ v_2]\right)\,\dd x\nn\\
&\geq&-\frac c\mu\left(\|v_1\|_2+\|v_1+v_2\|_2\right),
\eea
and argue  as before to arrive at a lower bound like that in (\ref{5.40}). Similarly, in the case $B=0$ but $A\neq0, C\neq0$, we can  estimate
\bea
\int_{\bfR^2}\left(\frac1{\lm_1}g_1 v_1+\frac1{\lm_2}g_2 v_2\right)\,\dd x&=&\int_{\bfR^2}\left(\left[\frac1{\lm_1}g_1+\frac1{\lm_2}g_2\right] v_1-\frac1{\lm_2}g_2[v_1- v_2]\right)\,\dd x\nn\\
&\geq&-\frac c\mu\left(\|v_1\|_2+\|v_1-v_2\|_2\right),
\eea
so that a similar conclusion holds. In view of these results, we see that (\ref{Iv12}) allows us to obtain some constants $a,b,c,C_0>0$ such that
the following coercive lower bound holds,
\bea \label{5.43}
I(v_1,v_2)&\geq& a\left(\|\nabla v_1\|^2_2+\|\nabla v_2\|^2_2\right)+b\left(\|{v_1}_-\|_2+\|{v_2}_-\|_2\right)\nn\\
&&+c\left(\|{v_1}_+\|_2^2+\|{v_2}_+\|_2^2\right)-C_0,\quad v_1,v_2\in W^{1,2}(\bfR^2).
\eea 

Therefore,  by virtue of (\ref{5.43}), we can use  standard arguments to obtain the existence of a 
global minimum point for $I$ as its unique critical point in $W^{1,2}(\bfR^2)\times W^{1,2}(\bfR^2)$, provided at least two of the constants
$A,B$, and $C$ do not vanish. Hence, the following slight extension of part (ii) of
 Theorem \ref{theorem2.2} holds: 

\begin{theorem} Consider the equations (\ref{4}) and (\ref{5}) over the full plane such that the vacuum constraint (\ref{2.7}) is
satisfied with at least two of the constants $A, B$, and $C$ not vanishing.
Then there always exists a unique solution $(u_1,u_2)$ satisfying the boundary condition $u_1,u_2\to0$ as $|x|\to\infty$.
\end{theorem}

\section{Proof of pointwise estimates}
\setcounter{equation}{0}

In this section, we analyze equations (\ref{4}) and (\ref{5}) under the vacuum constraint (\ref{2.7}) so that they take the forms:
\bea 
\Delta u_1&=&\lm_1\left(m|A|^2[\e^{mu_1}-1]+|B|^2[\e^{u_1+u_2}-1]+|C|^2[\e^{u_1-u_2}-1]\right)+4\pi\sum_{s=1}^{N_1}\delta_{p_{1,s}},\quad\label{8.1a}\\
\Delta u_2&=&\lm_2\left(|B|^2[\e^{u_1+u_2}-1]-|C|^2[\e^{u_1-u_2}-1]\right)+4\pi\sum_{s=1}^{N_2}\delta_{p_{2,s}}.\label{8.1b}
\eea
Actually, after some simple manipulation, it is easy to rewrite (\ref{8.1a})--(\ref{8.1b}) as follows:
\bea 
\Delta u_1&=&\lm_1\left(m|A|^2 [\e^{mu_1}-1]+[|B|^2+|C|^2][\e^{u_1}-1]+[|C|^2-|B|^2]\e^{u_1-u_2}[1-\e^{u_2}]\right.\nn\\
&&\left.\quad+ 2|B|^2 \e^{u_1}[\cosh u_2-1]\right)+4\pi\sum_{s=1}^{N_1}\delta_{p_{1,s}},\label{8.2a}\\
\Delta u_2&=&\lm_2\left(2|B|^2 \e^{u_1}\sinh u_2+[|C|^2-|B|^2]\e^{u_1-u_2}[\e^{u_2}-1]+[|C|^2-|B|^2][1-\e^{u_1}]\right)\nn\\
&&\quad+4\pi\sum_{s=1}^{N_2}\delta_{p_{2,s}}.\label{8.2b}
\eea
From (\ref{8.2a})--(\ref{8.2b}), by the maximum principle and with the aid of a continuity method, it is possible to show that for $|B|\leq |C|$ every
periodic solution $(u_1,u_2)$ of the equations satisfies $u_1<0$ and $u_2<0$.

Here we focus on the physically relevant situation where
\be\label{8.3}
|B|>|C|\quad\mbox{i.e.}\quad \xi_2>0.
\ee 
In the following, the solution $(u_1,u_2)$ considered will either be periodic or satisfies (\ref{8.2a})--(\ref{8.2b}) over $\bfR^2$ together with the boundary condition
(\ref{bc}).

Also notice that in the periodic situation considered over the cell domain $\Om$ and under the vacuum constraint (\ref{2.7}), the necessary and sufficient conditions for existence 
take the form:
\bea 
m|A|^2+2|B|^2&>&\frac{4\pi}{|\Om|}\left(\frac{N_1}{\lm_1}+\frac{N_2}{\lm_2}\right),\label{8.4a}\\
m|A|^2+2|C|^2&>&\frac{4\pi}{|\Om|}\left(\frac{N_1}{\lm_1}-\frac{N_2}{\lm_2}\right),\label{8.4b}
\eea
and provide explicit upper bounds on the vortex numbers $N_1$ and $N_2$ in terms of the size of $\Om$.
In addition, notice that (\ref{8.4a})--(\ref{8.4b}) allow $N_1=N_2=0$ (i.e., the absence of vortices) which yields only the trivial solution $u_1=u_2=0$, consistently
with the physical interpretation.

So from now on we assume $N_1>0$ and start by showing the following:
 
\begin{lemma}\label{lemma1}
Let $\lm_2>\lm_1$, $|B|>|C|$, and assume that
\be\label{8.5}
m^2|A|^2<2\left(\frac{\lm_2}{\lm_1}-1\right)|C|^2.
\ee
Then
\be \label{8.6}
\frac{u_1}{\lm_1}\pm\frac{u_2}{\lm_2}<0,\quad\mbox{and in particular}\quad u_1<0.
\ee 
\end{lemma}

\begin{remark} It is important to notice that our assumption (\ref{2.15x}) compares to (\ref{8.5}) as follows. If $0<m<2$, then (\ref{2.15x}) is stronger than
(\ref{8.5}) and implies it. If $m\geq2$, then (\ref{2.15x}) and (\ref{8.5}) are one and the same.
\end{remark}

\begin{proof} Observe that, if by contradiction we assume that there exists $x_0$:
\be 
\frac1{\lm_1} u_1(x_0)-\frac1{\lm_2} u_2(x_0)=\max\left(\frac1{\lm_1}u_1-\frac1{\lm_2} u_2\right)\geq0,
\ee
then $x_0\not\in\{p_{1,1},\cdots,p_{1,N_1}\}\cup\{p_{2,1},\cdots,p_{2,N_2}\}$ and it satisfies
\bea 
0&\geq&\Delta\left(\frac{u_1}{\lm_1}-\frac{u_2}{\lm_2}\right)(x_0)=m|A|^2\left(\e^{mu_1(x_0)}-1\right)+2|C|^2\left(\e^{u_1(x_0)-u_2(x_0)}-1\right)\nn\\
&\geq&\lm_1\left(m^2|A|^2+2|C|^2\right)\left(\frac{u_1(x_0)}{\lm_1}-\frac{u_2(x_0)}{\lm_2}\right)+\frac{\lm_1}{\lm_2}\left(m^2|A|^2-2|C|^2\left[\frac{\lm_2}{\lm_1}-1\right]\right)u_2(x_0).\nn
\eea
Thus, from (\ref{8.5}), we get $u_2(x_0)\geq0$. As a consequence, $u_1(x_0)\geq (\lm_1/\lm_2)u_2(x_0)\geq0$. Therefore, there exists $x_1\not\in\{p_{1,1},\cdots,p_{1,N_1}\}$ such that
\[
u_1(x_1)=\max u_1\geq0.
\]
But then from (\ref{8.2a}) we readily see that $u_2(x_1)\leq0$ and we conclude with
\[
\frac{u_1(x_1)}{\lm_1}\leq\frac{u_1(x_1)}{\lm_1}-\frac{u_2(x_1)}{\lm_2}\leq \frac{u_1(x_0)}{\lm_1}-\frac{u_2(x_0)}{\lm_2}\leq \frac{u_1(x_0)}{\lm_1}\leq\frac{u_1(x_1)}{\lm_1}.
\]
That is, $u_2(x_1)=u_2(x_0)=0$ and by virtue of (\ref{8.2a}) again we have
\[
\max u_1=u_1(x_1)=0=\frac{u_1(x_0)}{\lm_1}-\frac{u_2(x_0)}{\lm_2}=\max\left\{\frac{u_1}{\lm_1}-\frac{u_2}{\lm_2}\right\}.
\]
Consequently,
\bea 
\Delta\left(\frac{u_1}{\lm_1}-\frac{u_2}{\lm_2}\right)&\geq&\left(m^2|A|^2-2|C|^2\left[\frac{\lm_2}{\lm_1}-1\right]\right) u_1+2|C|^2\lm_2\left(\frac{u_1}{\lm_1}-\frac{u_2}{\lm_2}\right)\nn\\
&\geq& 2|C|^2\lm_2\left(\frac{u_1}{\lm_1}-\frac{u_2}{\lm_2}\right).\nn
\eea
So, we can use the strong maximum principle to conclude that 
\[
\frac{u_1}{\lm_1}\equiv \frac{u_2}{\lm_2},
\]
which is certainly never allowed by (\ref{8.1a})--(\ref{8.1b}), unless $u_1\equiv u_2\equiv0$. That is, $N_1=N_2=0$. Hence for $N_1>0$, there holds
\be\label{8.7}
\frac{u_1}{\lm_1}-\frac{u_2}{\lm_2}<0.
\ee

On the other hand, from (\ref{8.7}) we also get that, necessarily, $u_1<0$. Indeed, if, as before, we assume that there exists $x_1\not\in\{p_{1,s}\,|\,s=1,\cdots, N_1\}$ such that
$u_1(x_1)=\max u_1\geq0$, then $u_2(x_1)\leq0$ and we get a contradiction as follows:
\[
0\leq \frac{u_1(x_1)}{\lm_1}\leq\frac{u_1(x_1)}{\lm_1}-\frac{u_2(x_1)}{\lm_2}<0.
\]
In turn, we see that
\bea 
\Delta\left(\frac{u_1}{\lm_1}+\frac{u_2}{\lm_2}\right)&=&m|A|^2(\e^{mu_1}-1)+2|B|^2(\e^{u_1+u_2}-1)+\frac{4\pi}{\lm_1}\sum_{s=1}{N_1}\delta_{p_{1,s}}+\frac{4\pi}{\lm_2}\sum_{s=1}^{N_2}
\delta_{p_{2,s}}\nn\\
&\geq&\left(m^2|A|^2-2|B|^2\left[\frac{\lm_2}{\lm_1}-1\right]\right)u_1+2|B|^2\lm_2\left(\frac{u_1}{\lm_1}+\frac{u_2}{\lm_2}\right)\nn\\
&\geq&2|B|^2\lm_2\left(\frac{u_1}{\lm_1}+\frac{u_2}{\lm_2}\right),\nn
\eea
where the last inequality follows by the observation that
\[ 
m^2|A|^2-2|B|^2\left(\frac{\lm_2}{\lm_1}-1\right)<m^2|A|^2-2|C|^2\left(\frac{\lm_2}{\lm_1}-1\right)<0\quad\mbox{and}\quad u_1<0.
\]
Whence,
\[
\frac{u_1}{\lm_1}+\frac{u_2}{\lm_2}<0,
\]
as it follows easily by the maximum principle.
\end{proof}

Next we show how the stronger assumption (\ref{2.15x}) (when $0<m<2$) actually implies
\[
u_1\pm u_2<0.
\]

To this purpose, observe that
\bea 
\Delta(u_1+u_2)&=&\lm_1\left(m|A|^2[\e^{mu_1}-1]-\left[\frac{\lm_2}{\lm_1}-1\right]|C|^2[\e^{u_1-u_2}-1]\right)\nn\\
&&\,+(\lm_1+\lm_2)|B|^2(\e^{u_1+u_2}-1)
+4\pi\left(\sum_{s=1}^{N_1}\delta_{p_{1,s}}+\sum_{s=1}^{N_2}\delta_{p_{2,s}}\right),\label{8.8a}\\
\Delta(u_1-u_2)&=&\lm_1\left(m|A|^2[\e^{mu_1}-1]-\left[\frac{\lm_2}{\lm_1}-1\right]|B|^2[\e^{u_1+u_2}-1]\right)\nn\\
&&\,+(\lm_1+\lm_2)|C|^2(\e^{u_1-u_2}-1)
+4\pi\left(\sum_{s=1}^{N_1}\delta_{p_{1,s}}-\sum_{s=1}^{N_2}\delta_{p_{2,s}}\right).\label{8.8b}
\eea
Therefore, if by contradiction we assume that there exists
\[ 
\bar{x}\not\in\{p_{1,1},\cdots,p_{1,N_1}\}:\quad (u_1+u_2)(\bar{x})=\max(u_1+u_2)\geq0,
\]
then $u_1(\bar{x})\geq -u_2(\bar{x})$, and we find:
\[
0\geq\Delta(u_1+u_2)(\bar{x})\geq\lm_1\left(m|A|^2[\e^{mu_1(\bar{x})}-1]-\left[\frac{\lm_2}{\lm_1}-1\right]|C|^2[\e^{2u_1(\bar{x})}-1]\right).
\]
In case $0<m\leq2$, then realizing that $u_1<0$ (by Lemma \ref{lemma1}) and by (\ref{2.15x}), we obtain the desired contradiction as follows:
\[
m|A|^2(\e^{mu_1(\bar{x})}-1)-\left(\frac{\lm_2}{\lm_1}-1\right)|C|^2(\e^{2u_1(\bar{x})}-1)\geq \left(m|A|^2-\left[\frac{\lm_2}{\lm_1}-1\right]|C|^2\right)(\e^{2u_1(\bar{x})}-1)>0.
\]
In case $m>2$, then (\ref{8.5}) and (\ref{2.15x}) coincide and in this case, we arrive at the desired contradiction, by observing that under (\ref{8.5}), the function
\[
f(t)=m|A|^2(\e^{mt}-1)-\left(\frac{\lm_2}{\lm_1}-1\right)|C|^2(\e^{2t}-1)
\]
satisfies $f(t)>0$, $\forall t<0$.

Similarly, we show that $u_1-u_2<0$. Since if by contradiction we assume there exists 
\[
\hat{x}\not\in\{p_{1,1},\cdots,p_{1,N_1}\}\setminus\{p_{2,1},\cdots,p_{2,N_2}\}:\quad (u_1-u_2)(\hat{x})=\max(u_1-u_2)\geq0,
\]
then $u_1(\hat{x})\geq u_2(\hat{x})$ and we get
\[
0\geq\Delta(u_1-u_2)(\hat{x})\geq\lm_1\left(m|A|^2[\e^{mu_1(\hat{x})}-1]-\left[\frac{\lm_2}{\lm_1}-1\right]|B|^2[\e^{2u_1(\hat{x})}-1]\right).
\]
At this point, recalling that $|B|>|C|$, we reach a contradiction arguing exactly as above.
%%%%%%%%%%%%%%%%%

To conclude the proof of Theorem \ref{theorem2.2}, it  remains to establish (\ref{q1}) and (\ref{q2}). In the periodic case,
they just express (\ref{et1}) and (\ref{et2}) once we introduce the vacuum constraint (\ref{2.7}) and use (\ref{x2.16}).
In the planar case, they follow exactly as in the Abelian case (e.g., see the proof of Proposition
3.2.4 in \cite{Ta1}). Indeed, multiplying (\ref{eq1}) by the test function $\chi_R(x)=\chi(\frac xR)$ with $\chi\in C^\infty_c(B_2)$ satisfying $\chi\equiv1$ in $B_1$ 
 and integrate over $\bfR^2$, to obtain:
\bea 
&&\int_{\bfR^2}\left\{m|A|^2(1-\e^{mu_1})+|B|^2(1-\e^{u_1+u_2})+|C|^2 (1-\e^{u_1-u_2})\right\}\chi_R\,\dd x\nn\\
&&=\frac1{\lm_1}\left(\int_{\bfR^2} g_1\chi_R\,\dd x-\int_{\bfR^2}\Delta v_1\chi_R\,\dd x\right)\nn\\
&&=\frac1{\lm_1}\left(\int_{\bfR^2} g_1\chi_R\,\dd x-\int_{\bfR^2} v_1 \Delta\chi_R\,\dd x\right)\to 4\pi \frac{N_1}{\lm_1}\quad\mbox{as }R\to\infty,\nn
\eea
since
\[
\lim_{R\to\infty}\int_{\bfR^2} g_1 \chi_R\,\dd x=\int_{\bfR^2} g_1\,\dd x=4\pi N_1,
\]
and
\[ 
\left|\int_{\bfR^2} v_1\chi_R\,\dd x\right|\leq\left(\|v_1\|_{L^\infty(R\leq|x|\leq 2R)}\right)\left(\int_{B_2}|\Delta \chi|\,\dd x\right)\to0\quad \mbox{as }R\to\infty.
\]
Thus, we deduce that $(1-\e^{mu_1}),(1-\e^{u_1\pm u_2})\in L^1(\bfR^2)$ and
\be 
m|A|^2\|1-\e^{mu_1}\|_1+|B|^2\|1-\e^{u_1+u_2}\|_1+|C|^2\|1-\e^{u_1-u_2}\|_1=\frac{4\pi N_1}{\lm_1}.
\ee

Using (\ref{eq2}), by a similar argument, we get
\be 
|B|^2\|1-\e^{u_1+u_2}\|_1-|C|^2\|1-\e^{u_1-u_2}\|_1=\frac1{\lm_2}\int_{\bfR^2} g_2\,\dd x=\frac{4\pi N_2}{\lm_2}.
\ee

So (\ref{q1}) and (\ref{q2}) are established.

In concluding this section, we mention that, by standard methods, it is not hard to show that, under the given assumption, $u_1\pm u_2$ approach zero at infinity
exponentially fast.
\section{Degenerate cases: doubly periodic solutions}
\setcounter{equation}{0}

In this section, we focus on the problems over a doubly periodic domain $\Om$.

(i) $A\neq0, B=C=0$. In this case the system decouples and we reduce to analyze the single equation:
\be \label{ab}
\Delta u_1=\lm_1(m|A|^2 \e^{mu}-\xi_1)+4\pi\sum_{s=1}^{N_1}\delta_{p_{1,s}}(x),
\ee
which is the well-known Abelian vortex equation, that admits a unique solution if and only if
\be \label{abc}
4\pi\frac{N_1}{\lm_1}<\xi_1|\Om|.
\ee

(ii) $A=0, B\neq0, C=0$. The equations (\ref{2.2}) and (\ref{2.3}) now become
\bea 
\Delta v_1&=&\lm_1 \left(|B|^2\e^{u_1^0+u_2^0+v_1+v_2}-\xi_1\right)+\frac{4\pi}{|\Om|}{N_1},\quad\quad\label{2.15}\\
\Delta v_2&=&\lm_2 \left(|B|^2\e^{u^0_1+u^0_2+v_1+v_2}-\xi_2\right)+\frac{4\pi}{|\Om|}{N_2},\label{2.16}
\eea
and no longer we need to assume (\ref{7}). Since
 (\ref{et2}) gives us
\be\label{2.17}
\eta_2=0,
\ee
then 
\be 
\Delta \left(\frac{v_1}{\lm_1}-\frac{v_2}{\lm_2}\right)=-(\xi_1-\xi_2)+\frac{4\pi}{|\Om|}\left(\frac{N_1}{\lm_1}-\frac{N_2}{\lm_2}\right)=-\frac2{|\Om|}\eta_2=0,
\ee
 which implies that $v_2=\lm_1^{-1}\lm_2 v_1$ up to a constant. The equations are turned into a single one, and a unique solution exists if and only if (\ref{abc}) holds.

(iii) $A=0, B=0, C\neq0$. The governing equations now read
\bea 
\Delta v_1&=&\lm_1 \left(|C|^2\e^{u_3^0+v_1-v_2}-\xi_1\right)+\frac{4\pi}{|\Om|}{N_1},\quad\quad\label{2.19}\\
\Delta v_2&=&-\lm_2 \left(|C|^2\e^{u_3^0+v_1-v_2}+\xi_2\right)+\frac{4\pi}{|\Om|}{N_2}.\label{2.20}
\eea
Since (\ref{et1}) implies $\eta_1=0$ or
\[
(\xi_1+\xi_2)|\Om|-4\pi\left(\frac{N_1}{\lm_1}+\frac{N_2}{\lm_2}\right)=0,
\]
 we have
\be 
\Delta\left(\frac{v_1}{\lm_1}+\frac{v_2}{\lm_2}\right)=0
\ee
and, thus, $v_2=-\lm^{-1}_1\lm_2 v_1$ up to a constant, and we arrive at the same conclusion as in the case (ii).

(iv) $A\neq0, B\neq0, C=0$. Again, we no longer need to require (\ref{7}), as the equations are
\bea 
\Delta v_1&=&\lm_1 \left(m|A|^2\e^{mu_1^0+mv_1}+|B|^2\e^{u_1^0+u_2^0+v_1+v_2}-\xi_1\right)+\frac{4\pi}{|\Om|}{N_1},\quad\quad\label{2.22}\\
\Delta v_2&=&\lm_2 \left(|B|^2\e^{u^0_1+u^0_2+v_1+v_2}-\xi_2\right)+\frac{4\pi}{|\Om|}{N_2},\label{2.23}
\eea
with associated action functional
\bea \label{I3}
I(v_1,v_2)
&=&\int_\Om\left\{\frac1{2\lm_1}|\nabla v_1|^2+\frac1{2\lm_2}|\nabla v_2|^2+|A|^2\e^{mu_1^0+mv_1}+|B|^2\e^{u_1^0+u_2^0+v_1+v_2}\right\}\,\dd x\nn\\
&&\quad -\eta_1(\ud{v}_1+\ud{v}_2)-\eta_2(\ud{v}_1-\ud{v}_2).
\eea
On the other hand, with $C=0$, the constraints (\ref{et1}) and (\ref{et2}) are refined into
\bea 
\int_\Om \frac m2|A|^2\e^{mu_1^0+mv_1}\,\dd x&=&\eta_2>0,\\
\int_\Om |B|^2\e^{u_1^0+u_2^0+v_1+v_2}\,\dd x&=&\eta_1-\eta_2\equiv\eta_3>0.
\eea
Thus, we can substitute
\be 
\eta_1(\ud{v}_1+\ud{v}_2)+\eta_2(\ud{v}_1-\ud{v}_2)=2\eta_2 \ud{v}_1+\eta_3 (\ud{v}_1+\ud{v}_2)
\ee
into (\ref{I3}) and follow the proof of Theorem \ref{theorem2.2} to show that the functional (\ref{I3}) has a unique critical point in $H(\Om)\times H(\Om)$.

(v) $A\neq0, B=0, C\neq0$. The equations are 
\bea 
\Delta v_1&=&\lm_1 \left(m|A|^2\e^{mu_1^0+mv_1}+|C|^2\e^{u_3^0+v_1-v_2}-\xi_1\right)+\frac{4\pi}{|\Om|}{N_1},\quad\quad\label{2.27}\\
\Delta v_2&=&-\lm_2 \left(|C|^2\e^{u_3^0+v_1-v_2}+\xi_2\right)+\frac{4\pi}{|\Om|}{N_2},\label{2.28}
\eea
with associated action functional
\bea \label{I4}
I(v_1,v_2)
&=&\int_\Om\left\{\frac1{2\lm_1}|\nabla v_1|^2+\frac1{2\lm_2}|\nabla v_2|^2+|A|^2\e^{mu_1^0+mv_1}+|C|^2\e^{u_3^0+v_1-v_2}\right\}\,\dd x\nn\\
&&\quad -\eta_1(\ud{v}_1+\ud{v}_2)-\eta_2(\ud{v}_1-\ud{v}_2).
\eea
When $B=0$, the constraints (\ref{et1}) and (\ref{et2}) become
\bea 
\int_\Om \frac m2|A|^2\e^{mu_1^0+mv_1}\,\dd x&=&\eta_1>0,\\
\int_{\Om}|C|^2\e^{u_3^0+v_1-v_2}\,\dd x&=&
\eta_2-\eta_1\equiv\eta_4>0.
\eea
Therefore, we may insert
\be 
\eta_1(\ud{v}_1+\ud{v}_2)+\eta_2(\ud{v}_1-\ud{v}_2)=2\eta_1 \ud{v}_1+\eta_4 (\ud{v}_1-\ud{v}_2)
\ee
into (\ref{I4}) and adapt the proof of Theorem \ref{theorem2.2} to establish the existence and uniqueness of a critical point of the functional (\ref{I4}).

In summary, we can state

\begin{theorem}\label{theorem2}
In the degenerate cases when $B$ or $C$ vanishes (we need to require
(\ref{7}) only when $C\neq0$), the existence and uniqueness of a solution to the master equations (\ref{4}) and (\ref{5}) hold under the following
necessary and sufficient conditions.

If $A\neq0, B=C=0$, the equations are reduced to (\ref{ab}) and the condition is (\ref{abc}). If $A=0, B\neq0, C=0$ or $A=0, B=0, C\neq0$, the equations are also reduced
to (\ref{ab}) and the condition reads
\be 
\frac{N_1}{\lm_1}<\xi_1\frac{|\Om|}{4\pi},
\ee
 along with
\be \label{6.21}
\frac{N_1}{\lm_1}-\frac{N_2}{\lm_2}=(\xi_1-\xi_2)\frac{|\Om|}{4\pi},
\ee
 or 
\be \label{6.22}
\frac{N_1}{\lm_1}+\frac{N_2}{\lm_2}=(\xi_1+\xi_2)\frac{|\Om|}{4\pi},
\ee
respectively.

If $A\neq0, B\neq0, C=0$, then the condition reads
\be 
4\pi\left(\frac{N_1}{\lm_1}-\frac{N_2}{\lm_2}\right)<(\xi_1-\xi_2)|\Om|,\quad 4\pi\frac{N_2}{\lm_2}<\xi_2|\Om|.\label{cd3}
\ee

If $A\neq0, B=0, C\neq0$, then the condition states
\be
4\pi\left(\frac{N_1}{\lm_1}+\frac{N_2}{\lm_2}\right)<(\xi_1+\xi_2)|\Om|,\quad \xi_2|\Om|<4\pi\frac{N_2}{\lm_2}. \label{cd4}
\ee

In all cases, the solutions may be constructed by a direct minimization method.
\end{theorem}

It is interesting to note that the pair of conditions (\ref{6.21}) and (\ref{6.22}), and (\ref{cd3}) and (\ref{cd4}), are interchangeable under the correspondence $(N_2,\xi_2)\leftrightarrow -(N_2,\xi_2)$.

\section{Degenerate cases: planar solutions}
\setcounter{equation}{0}

We now consider the degenerate cases concerning the planar solutions of the equations (\ref{4}) and (\ref{5}) subject to the boundary conditions (\ref{bc}).
We have already shown existence and uniqueness when at least two of the constants $A$, $B$, and $C$ do not vanish. Here we consider the other degenerate cases.

(i) $A\neq0, B=C=0$. In this case, we do not need (\ref{7}), and the condition (\ref{2.7}) gives $\xi_2=0$ and the system decouples into
an equation for $u_1$ as that in the Abelian Higgs model (\ref{a1.6}) and a linear equation for $u_2$, given as follows:
\bea 
\Delta u_1&=&\lm_1\left(m|A|^2\e^{m u_1}-\xi_1\right)+4\pi\sum_{s=1}^{N_1}\delta_{p_{1,s}}(x),\label{eq8.1}\\
\Delta u_2&=&4\pi\sum_{s=1}^{N_2}\delta_{p_{2,s}}(x).\label{eq8.2}
\eea
Obviously, the natural boundary condition for (\ref{eq8.1}), (\ref{eq8.2}) only involves $u_1$ and requires $u_1(x)\to 0$ as $|x|\to\infty$. Thus the existence and uniqueness
of a solution $u_1<0$ follows by the same analysis of the Abelian Higgs model (\ref{a1.6}). While a solution of (\ref{eq8.2}) with logarithmic growth may be obtained via the
fundamental solution.

(ii) $A=0, B\neq0, C=0$. From (\ref{2.7}), we see that $\xi_1=\xi_2\equiv\xi=|B|^2$. Thus the system becomes
\bea
\Delta u_1&=&\lm_1\left(|B|^2\e^{u_1+u_2}-\xi\right)+4\pi\sum_{s=1}^{N_1}\delta_{p_{1,s}}(x),\label{x7.3}\\
\Delta u_2&=&\lm_2\left(|B|^2\e^{u_1+u_2}-\xi\right)+4\pi\sum_{s=1}^{N_2}\delta_{p_{2,s}}(x).\label{x7.4}
\eea
Now the natural boundary condition involves only $u_1+u_2$ as follows: $(u_1+u_2)(x)\to0$ as $|x|\to\infty$.
With the background functions defined in (\ref{5u}) and (\ref{5g}) and setting $u_j=u_j^0+v_j$ ($j=1,2$), we may recast the above equations into
\bea
\Delta v_1&=&\lm_1\left(|B|^2\e^{u_1^0+u_2^0+v_1+v_2}-\xi\right)+g_1,\label{5v1}\\
\Delta v_2&=&\lm_2\left(|B|^2\e^{u_1^0+u_2^0+v_1+v_2}-\xi\right)+g_2.\label{5v2}
\eea
%It is easily seen that $|\nabla v_1|$ and $|\nabla v_2|$ vanish at infinity as rapidly as $\mbox{O}(|x|^{-3})$. Hence $\int_{\bfR^2}\Delta v_1\,\dd x=\int_{\bfR^2}\Delta v_2\,\dd x=0$.
%Besides, it may be checked directly that $\int_{\bfR^2} g_1\,\dd x=4\pi N_1, \int_{\bfR^2}g_2\,\dd x=4\pi N_2$. Integrating (\ref{5v1}) and (\ref{5v2}) and inserting the above results, we arrive
%at
%\be 
%\frac{N_1}{\lm_1}=\frac{N_2}{\lm_2}.
%\ee 
%Since $N_2\leq N_1$, we have $\lm_1\geq \lm_2$ and, with $w_1=u_1+u_2, w_2=u_1-u_2$, we arrive at the lower-triangular system
%\bea 
%\Delta w_1&=&(\lm_1+\lm_2)\left(|B|^2\e^{w_1}-\xi\right)+4\pi\left(\sum_{s=1}^{N_1}\delta_{p_{1,s}}(x)+
%\sum_{s=1}^{N_2}\delta_{p_{2,s}}(x)\right),\label{7.6}\\
%\Delta w_2&=&(\lm_1-\lm_2)\left(|B|^2\e^{w_1}-\xi\right)+4\pi\sum_{s={N_2}+1}^{N_1}\delta_{p_{1,s}}(x),
%\eea 
%if $N_2<N_1$ (if $N_2=N_1$, then $\xi_1=\xi_2$ and the equation for $w_2$ is $\Delta w_2=0$ which is trivially fulfilled with the choice $w_2\equiv0$). The key component
%of the system, (\ref{7.6}),
%is again the equation for the Abelian Higgs model.

%(iii) $A=B=0, C\neq0$. From (\ref{2.7}), we can read off $\xi_2=-\xi_1<0$ and deduce $\lm_2<0$. Thus, we are beyond the interested range of the parameters.

%(iv) The cases $A\neq0, B\neq0, C=0$ and $A\neq0, B=0, C\neq 0$ are already treated in the previous section.

%In summary, we can slightly extend the planar part of Theorem \ref{theorem2.2} as follows.

Again, we no longer need to require (\ref{7}). Moreover, (\ref{5v1}) and (\ref{5v2}) can be easily handled as the Abelian model (\ref{a1.6}) once we
introduce the known $w=v_1+v_2$ that must satisfy:
\be \label{x7.7}
\left.\begin{array}{ll}
&\Delta w=(\lm_1+\lm_2)|B|^2 (H_1\e^w-1)+g_1+g_2,\\
&w\to0\quad\mbox{as }|x|\to\infty,\end{array}\right.
\ee
where the function $H_1$ is as given in (\ref{x5.3}).

We know the existence and uniqueness of a solution $w\in W^{1,2}(\bfR^2)$ of (\ref{x7.7}) that also satisfies:
$H_1 \e^w<1$, $w\to0$ exponentially fast at infinity, and
\be \label{x7.8}
|B|^2\int_{\bfR^2}(1-H_1 \e^w)\,\dd x=\frac{4\pi (N_1+N_2)}{\lm_1+\lm_2}.
\ee
At this point, we can use $v_1+v_2=w$ in (\ref{5v1})--(\ref{5v2}) and linear elliptic theory to find $v_1$ and $v_2$ (unique up to an additive constant) and see that
in general they admit a logarithmic growth at infinity at a rate determined by the integral value of the right-hand side of (\ref{5v1}) and (\ref{5v2}) respectively.

From (\ref{x7.8}), we easily compute:
\bea
\frac1{2\pi}\left(\lm_1|B|^2\int_{\bfR^2}(H_1\e^w-1)\,\dd x+4\pi N_1\right)&=&\frac{2\lm_1\lm_2}{\lm_1+\lm_2}\left(\frac{N_1}{\lm_1}-\frac{N_2}{\lm_2}\right)\nn\\
&=&-\frac1{2\pi}\left(\lm_2|B|^2\int_{\bfR^2}(H_1\e^w-1)\,\dd x+4\pi N_2\right),\nn
\eea
and obtain:
\be\label{x7.9}
\begin{array}{ll}&v_1(x)=\frac{2\lm_1\lm_2}{\lm_1+\lm_2}\left(\frac{N_1}{\lm_1}-\frac{N_2}{\lm_2}\right)\ln|x|+\mbox{O}(1),\\
&v_2(x)=-\frac{2\lm_1\lm_2}{\lm_1+\lm_2}\left(\frac{N_1}{\lm_1}-\frac{N_2}{\lm_2}\right)\ln|x|+\mbox{O}(1),\end{array}
\mbox{as }|x|\to\infty.
\ee
For details, see e.g. Lemma 1.1 and Corollary 1.1 of \cite{CT}.

Thus, in this case in terms of the original variables $u_1$ and $u_2$, we can claim the existence of a solution for (\ref{x7.3})--(\ref{x7.4})
satisfying: $u_1+u_2<0$, $u_1+u_2\to 0$ exponentially fast at infinity, and
\[
|B|^2\int_{\bfR^2}(1-\e^{u_1+u_2})\,\dd x=\frac{4\pi(N_1+N_2)}{\lm_1+\lm_2}.
\]
Moreover, by the maximum principle, we see that:
\[ 
\mbox{if }\frac{N_1}{\lm_1}-\frac{N_2}{\lm_2}<0,\quad\mbox{then }u_1<0\mbox{ in }\bfR^2,
\]
\[ 
\mbox{if }\frac{N_1}{\lm_1}-\frac{N_2}{\lm_2}>0,\quad\mbox{then }u_2<0\mbox{ in }\bfR^2.
\]
Only in the limiting case: $\frac{N_1}{\lm_1}=\frac{N_2}{\lm_1}$, we can actually ensure that $(v_1,v_2)\in W^{1,2}(\bfR^2)\times W^{1,2}(\bfR^2)$; and so both
$u_1$ and $u_2$ vanish at infinity and $u_1<0, u_2<0$ in $\bfR^2$.

(iii) $A=B=0$, $C\neq0$. This case requires (\ref{7}) and the natural boundary condition: $(u_1-u_2)(x)\to0$ as $|x|\to\infty$.

Such a case can be treated as above with the help of the new variable $w=v_1-v_2$, which must satisfy a problem as (\ref{x7.7}) only with $H_1$ replaced by $H_2$, $|B|^2$
by $|C|^2$, and $g_1+g_2$ by $g_1-g_2$.

So, by similar arguments, also in this case the existence of a solution pair $(u_1,u_2)$ can be established satisfying:
$u_1-u_2<0$ in $\bfR^2$, $u_1-u_2\to0$ at infinity exponentially fast, and
\[ 
|C|^2\int_{\bfR^2}(1-\e^{u_1-u_2})\,\dd x=\frac{4\pi(N_1-N_2)}{\lm_1+\lm_2}.
\]

%By standard methods, it is not hard to show that all the obtained planar solutions of the master equation (\ref{1.1}) and 
%the system of the master equations (\ref{4}) and (\ref{5}) vanish at infinity exponentially fast. We omit the details.

Finally, we remark that in all parameter regimes the following ``flux quantization" formulas
\be 
\int\left( m|A|^2\e^{mu_1}+|B|^2\e^{u_1+u_2}+|C|^2\e^{u_1-u_2}-\xi_1\right)=-\frac{4\pi N_1}{\lm_1},
\ee 
\be
\int\left(|B|^2\e^{u_1+u_2}-|C|^2 \e^{u_1-u_2}-\xi_2\right)=-\frac{4\pi N_2}{\lm_2},
\ee 
are valid for solutions of (\ref{4})--(\ref{5}) over a doubly periodic domain, or the full plane. These expressions are weaker than but cover those
stated in (\ref{q1}) and (\ref{q2}) which involve the $L^1$-norms of various quantities. They are exactly what have been used in literature for the computation of the tension of the non-Abelian
vortex tubes.

\small{

}
\end{document}